\documentclass[journal]{IEEEtran}
\usepackage{amsmath,amsfonts}
\usepackage{algorithmic}
\usepackage{algorithm}
\usepackage{array}
\usepackage[caption=false,font=normalsize,labelfont=sf,textfont=sf]{subfig}
\usepackage{textcomp}
\usepackage{stfloats}
\usepackage{url}
\usepackage{verbatim}
\usepackage{graphicx}
\usepackage{cite}

\usepackage{booktabs}
\usepackage{multirow}
\usepackage{hyperref}

\newtheorem{theorem}{Theorem}
\newtheorem{lemma}{Lemma}
\newtheorem{definition}{Definition}
\newtheorem{proof}{Proof}[section]

\begin{document}

\title{Core-periphery Detection Based on Masked Bayesian Non-negative Matrix Factorization}

\author{Zhonghao Wang, Ru Yuan, Jiaye Fu, Ka-Chun Wong, Chengbin Peng
	
	\thanks{Zhonghao Wang, Ru Yuan and Chengbin Peng are with College of Information Science and Engineering, Ningbo University, Ningbo 315200, China, and also with the Key Laboratory of Mobile Network Application Technology of Zhejiang Province, Ningbo 315200, China (e-mail: wongzhonghao123@gmail.com; shanshuiqiankun@gmail.com; pengchengbin@nbu.edu.cn).}
	\thanks{Ka-Chun Wong is with Department of Computer Science, City University of Hong Kong, Hong Kong 999077, China (email: 	kc.w@cityu.edu.hk).}
	\thanks{Jiaye Fu is with School of Management, Xi'an Polytechnic University, Xi'an 710048, China (email: 3220529324@qq.com).}
}

\markboth{}%
{Shell \MakeLowercase{\textit{et al.}}: A Sample Article Using IEEEtran.cls for IEEE Journals}

\maketitle

\begin{abstract}
	
Core-periphery structure is an essential mesoscale feature in complex networks. Previous researches mostly focus on discriminative approaches while in this work, we propose a generative model called masked Bayesian non-negative matrix factorization. We build the model using two pair affiliation matrices to indicate core-periphery pair associattions and using a mask matrix to highlight connections to core nodes. We propose an approach to infer the model parameters, and prove the convergence of variables with our approach. Besides the abilities as traditional approaches, it is able to identify core scores with overlapping core-periphery pairs. We verify the effectiveness of our method using randomly generated networks and real-world networks. Experimental results demonstrate that the proposed method outperforms traditional approaches. 
\end{abstract}


\section{Introduction}

\IEEEPARstart{C}omplex networks are frequently utilized to represent real-world systems in various fields, including social relations\cite{social}, biological interactions\cite{social_biological}, and brain networks\cite{brain_networks}, can be modeled as complex networks.  The investigation of complex network topology mainly focused on the global, mesoscale, and local structure of the network. Community structure\cite{newman2004finding}, core-periphery structure\cite{borgatti2000models}, and hierarchical structure\cite{peixoto2014hierarchical} are typical types of mesoscale structures. 

The core-periphery structure is a distinct form of community structure in which there are two partitions: the core nodes and the periphery nodes. 
In a same core-periphery pair, core nodes are densely connected while the connection between periphery nodes are sparse.
Figure.\ref{CP-structure} illustrates the core-periphery structure,
where yellow and blue nodes denote core and periphery nodes, respectively. As shown in the figure, each network may have several core-periphery pairs, and each pair may containing multiple core nodes. The concept of the core-periphery structure was first proposed by  Krugman et al. for economic analysis\cite{krugman1991increasing}, and then formalized by Borgatti and Everett \cite{borgatti2000models}. Recent studies have revealed many applications in analyzing collaboration networks \cite{CPcollab}, economic networks\cite{CPeconomic}, traffic networks \cite{CPtraffic}, word networks\cite{sarkar2022core}, and trading networks\cite{CPtrade}.
Some recent methods distinguish the core and periphery nodes by binary classification\cite{kojaku2017finding,brusco2011exact,zhang2015identification}. 
Some others measure a quantitative likelihood that each node is a core node 
\cite{borgatti2000models,minres,yan2019multicores,da2008centrality}.

Many approaches  have been devised for core-periphery detection. 
For example, Shen et al.\cite{shen2021finding} consider the core–periphery detection as a likelihood maximization problem, and proposed the C–P score maximization algorithm to detect core-periphery pairs. Jia et al.\cite{jia2019random} proposed a core score inference algorithm via likelihood maximization. Zhang et al.\cite{zhang2015identification} proposed a expectation–maximization algorithm to infer the parameter of stochastic block model of core–periphery structure. 
These approaches generally  estimate model parameters by maximizing a conditional probabilities based on given networks.
Nevertheless, these approaches suffer from some limitations particularly in addressing overlapping core-periphery pairs and in providing a comprehensive generative understanding of network with core-periphery structures.

Different from these approaches, in this work, we propose a generative model  \cite{jebara2012machine} 
that can approximately identify how a given network is generated with a few basic parameters and assumptions. The solution to this model can distinguish not only traditional  core-periphery pairs, but also overlapping ones, with relatively high accuracy. 
Our contributions can be summarized as follows:

\begin{itemize}
	\item %
	We propose a novel generative model of masked Bayesian non-negative matrix factorization, which is able to predict the likelihood of core scores and core-periphery pair affiliations.  %
	\item We theoretically prove that our approach can converge and demonstrate that it is applicable to overlapping core-periphery pairs. %
	\item %
	We verify the effectiveness of our method on synthetic networks and real-world networks, and demonstrate that it can be  accelerated easily with GPUs  to achieve a remarkable speedup. %
\end{itemize}

\begin{figure}[!t]
	\centering %
	\includegraphics[width=0.5\textwidth]{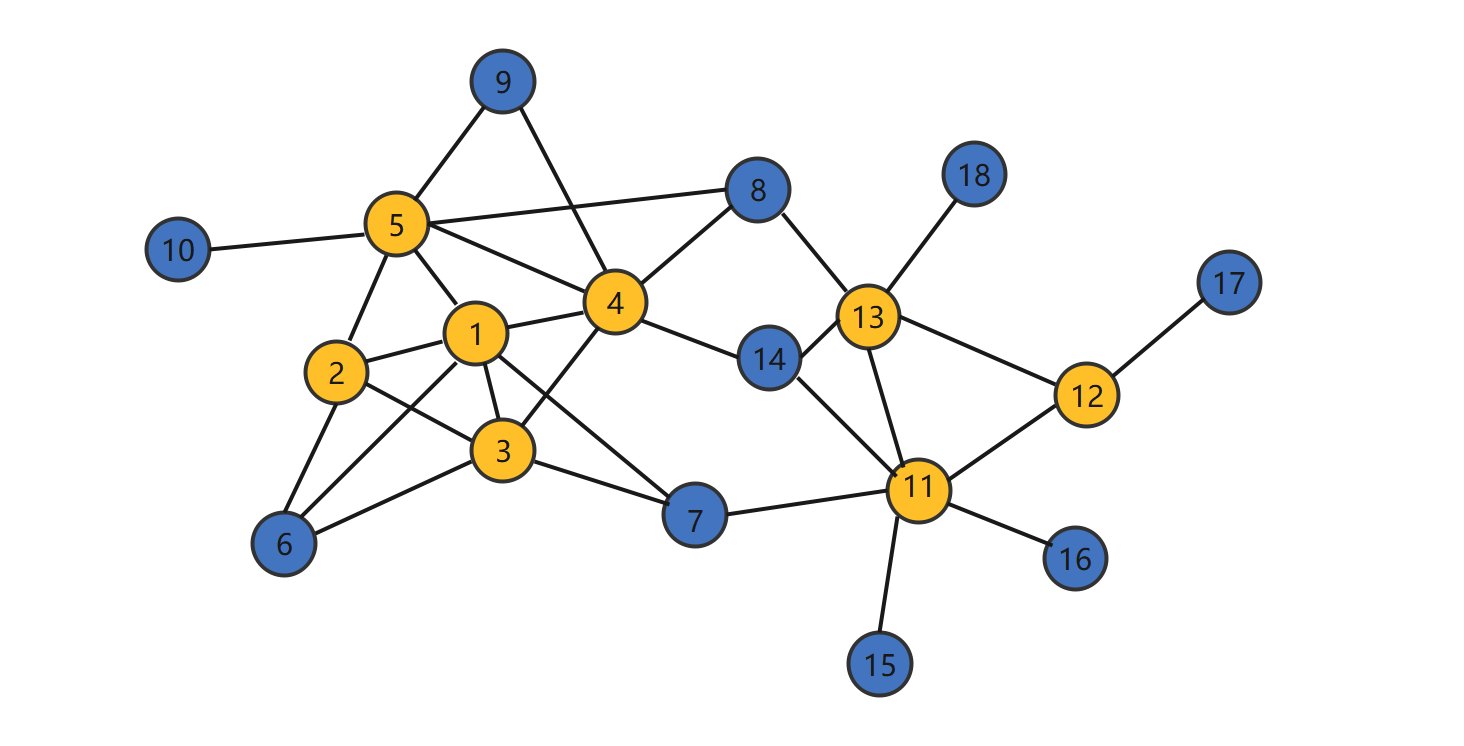} %
	\caption{A visual illustration of a simple network with two core-periphery structures. Yellow nodes are core nodes ,while blue nodes are periphery nodes.} %
	\label{CP-structure} %
\end{figure}%

\section{related work}

\subsection{Core-periphery detection}

Core-periphery detection has been investigated for decades, and two kinds of approaches have been developed. 
One kind of approach is based on binary classification. Kojaku et al. \cite{kojaku2017finding} propose a scalable algorithm to detect multiple nonoverlapping core-periphery pairs in a network, which extends the idea of core-periphery structure \cite{borgatti2000models}. It can also identify the number and size of core-periphery structures automatically. %
Zhang et al. \cite{zhang2015identification} propose a %
statistical inference method using expectation maximization and belief programs to fit a generative model to observed network data. %
Shen et al.\cite{shen2021finding} propose a metric to measure the performance of core-periphery detection algorithms, and propose a likelihood model to find the best solution in terms of that metric. %
Xiang et al.\cite{xiang2018unified} developed a unified framework for detecting core-peripheral structures and overlapping communities.
Ma et al.\cite{ma2018detection} proposed an parameter-free algorithm to detect the core-periphery structures based on the 3-tuple motif. 
Multi-class classification \cite{rezaei2020ml} or feature selection  \cite{karami2023unsupervised} methods may also help  for the analysis.

Another type of core-periphery detection approach is based on soft thresholds.
Yan et al. \cite{yan2019multicores} propose identifying multiple cores-periphery pairs through hierarchical clustering and using a difference score between empirical and random networks to choose the best partition. 
Boyd et al.\cite{minres} propose a method based on minimum residual singular value decomposition, which is suitable for diagonal missing or asymmetrical networks. 
Lee \cite{lee2014density} proposed a method based on density and transport and illustrated its usefulness in transportation networks.
Liu et al.\cite{9718598} proposed a hybrid method based on K-shell decomposition to identify the most influential spreaders in complex networks.
Shen et al.\cite{shen2021influences} proposed influence-based core-periphery detection approach to find multiple pairs of core-periphery nodes.

\subsection{Non-negative matrix factorization}
Non-negative matrix factorization (NMF) has numerous applications, including email surveillance\cite{berry2005email} and document clustering \cite{xu2003document}. The fundamental principle of NMF involves factorizing the given matrix into two non-negative matrices. 
Lee and Seung \cite{lee2000algorithms} proposed two different multiplicative algorithms for NMF, one algorithm aims to minimize the conventional least squares error, and the other algorithm aims to minimize the generalized Kullback-Leibler divergence. 
Gonzalez and Zhang \cite{gonzalez2005accelerating} further developed a variation of one of the Lee-Seung algorithms with a improved performance. 
Zdunek and Cichocki\cite{zdunek2006non} proposed a quasi-Newton method for NMF by considering the special structure of the Hessian of the Amari alpha divergence. 

In the context of community detection, the basic idea is factorizing the adjacency matrix of the observed network into two non-negative matrices, these matrices can be used to represent the importance of different nodes in different communities. In recent years, many NMF-based community detection approaches have been proposed.
Wang et al. \cite{wang2011community} proposed three NMF-based community detection techniques, namely Symmetric NMF, Asymmetric NMF, and Joint NMF, which can effectively detect community structures.
Shi et al. \cite{shi2015community} proposed the pairwisely constrained nonnegative symmetric matrix factorization (PCSNMF) method, which identifies community structures by considering both symmetric community structures of undirected network and pairwise constraints generated from some ground-truth group information. 
Psorakis et al. \cite{psorakis2011overlapping} proposed a community detection method based on Bayesian non-negative matrix factorization that can represent the overlapping between different communities and achieve soft community partitioning. 
Kamuhanda et al. \cite{9146784} proposed the Sparse Nonnegative Matrix Factorization (SNMF) for detecting multiple local communities.

However, none of these approaches focus on the detection of core-periphery structures. Our work further developed a core-periphery detection method based on masked Bayesian Non-negative Matrix Factorization. In the following section, we describe our approach in detail.

\begin{figure}[!t]
	\centering %
	\includegraphics[width=0.45\textwidth]{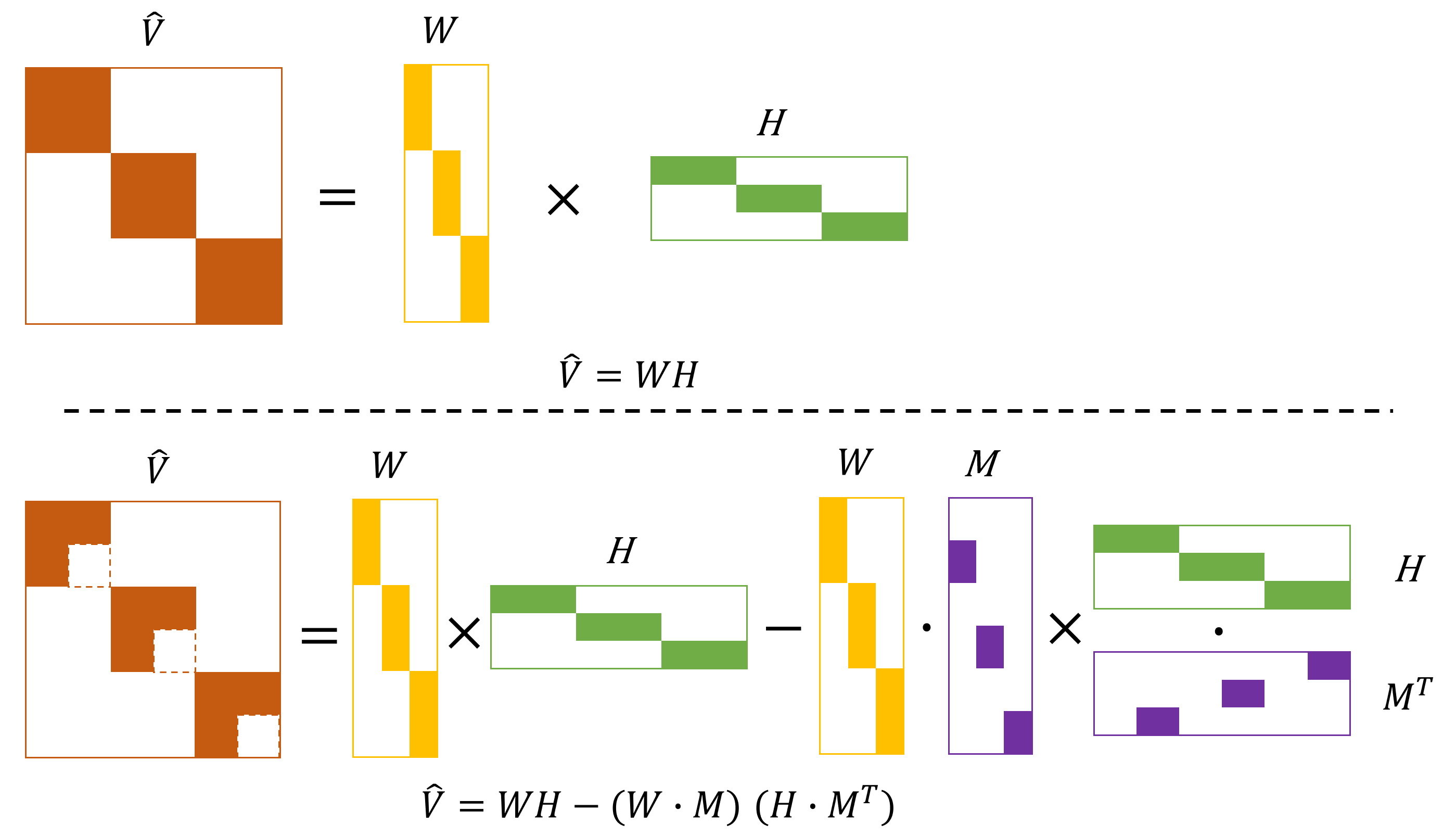} %
	\caption{Comparison between the traditional non-negative matrix factorization and the proposed factorization method. %
	$\hat{V}$ is an approximation to the adjacency matrix, %
	Factors $W$ and $H$ can indicate core-periphery pair affiliations. In our proposed method, a mask matrix $M$ is multiplied with $W$ and $H$ to highlight core nodes. %
    } %
	\label{CP-NMF} %
\end{figure}%

\section{method}
\subsection{Masked Bayesian non-negative matrix factorization}

 The basic idea of NMF is to decompose a matrix into two matrices, and entries of all these matrices are non-negative.  Formally, 
given a non-negative matrix $V \in \mathbb{R}^{N \times N}$, an NMF algorithm attempts to find two matrices $W \in \mathbb{R}^{N \times K}$ and $H \in \mathbb{R}^{K \times N}$ such that %

\begin{align}
	 V \approx \hat{V} = WH.
\end{align}

For core-periphery detection, we introduce a mask matrix  $M\in [0,1]^{N\times K}$ to highlight core nodes 
as follows: 

\begin{equation}
	V \approx \hat{V} = WH - (W\circ M) 
	 (H\circ M^T).\label{CP-NMF-EQ}
\end{equation}

Here, the notation '$\circ$' indicates element-wise matrix multiplication. In our analysis, matrix $V$ represents an adjacency matrix. The non-negative factor $W$ can be the same as $H^T$ when $V$ is symmetric, and $W_{ik}$ represents how likely a node belongs to the core-periphery pair $k$. A mask matrix $M$ can reduce the connection probability between periphery nodes by subtracting $(W \circ M)(H \circ M^T)$ from $WH$. 

As illustrated in Figure \ref{CP-NMF}, a traditional non-negative matrix factorization method cannot describe core-periphery structures, while our proposed approach works. 
 
After factorization, a larger $M$ indicates that the corresponding node is more likely to be a periphery node and vice versa. %
The matrix $M$ is functionally equivalent to a mask during the factorization. %

\begin{table}[!t]  \caption{Notations}
	\label{Notations}
	\centering
	\begin{tabular}{c|c}
		
		\toprule
		Notation  &  Explanation \\
		\midrule
		$N$  & Network size\\ 
		$K$  & Number of latent core-periphery pairs\\ 
		$V \in \mathbb{R}^{N \times N}$  & Adjacency matrix of observed network\\ 
		$\hat{V} \in \mathbb{R}^{N \times N}$  & Adjacency matrix of expectation network\\ 
		$W \in \mathbb{R}^{N \times K}$  & Non-negative interaction matrix\\ 
		$H \in \mathbb{R}^{K \times N}$  & Non-negative interaction matrix\\ 
		$M \in \mathbb{R}^{N \times K}$  & Non-negative mask matrix\\ 
		$\beta$  & Hyperparameter of the distribution of $W$ and $H$\\ 
		$\mu, \overline{\sigma}$  & Hyperparameters of the distribution of $M$\\
		$a,b$  & Hyperparameters of the distribution of $\beta$ \\ 
		$\hat{\mu}, \hat{\sigma}$  & Hyperparameters of the distribution of $\mu$\\
		\bottomrule  
	\end{tabular}
\end{table}

 \subsection{Likelihood model}
\label{likelihood-model}

\begin{figure}[!t]
	\centering %
	\includegraphics[width=0.45\textwidth]{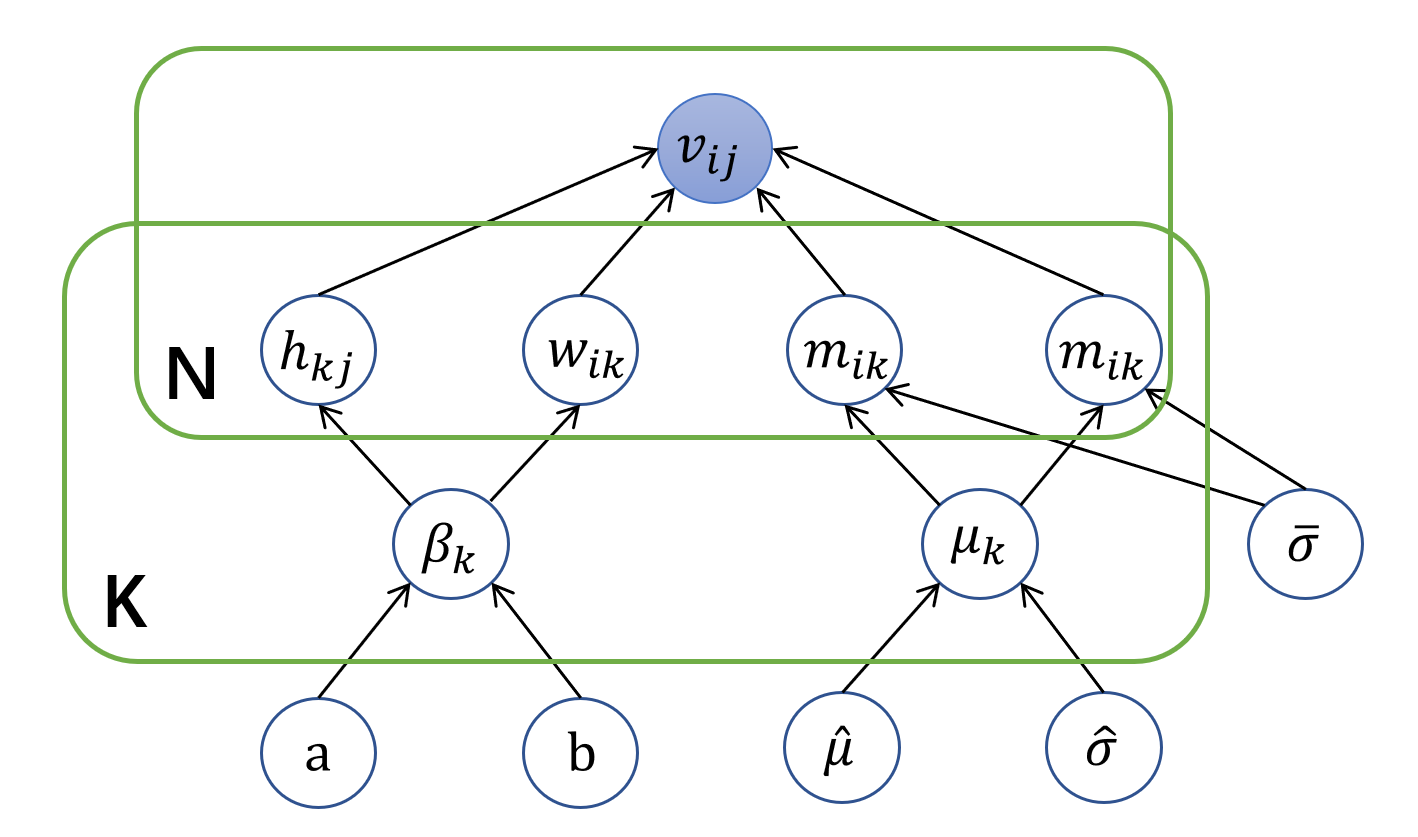} %
	\caption{A graphical illustration of our masked Bayesian non-negative matrix factorization model. %
		The observed value $V$ depends on %
		$W$, $H$, and $M$, and these variables further depend on  $\beta$ and  $\mu$.  $a$, $b$, $\overline{\sigma}$, $\hat{\mu}$, and $\hat{sigma}$ are hyper-parameters. %
	}
	\label{BayesianNMF} %
\end{figure}%

The overall architecture of our proposed Bayesian model is as shown in Figure \ref{BayesianNMF}. %

For ease of analysis, we assume that each entry of the adjacency matrix $V$ follows a Poisson distribution with parameter $\hat{V}$ as follows:
 \begin{align}
 	p(V|\hat{V}) = p(V|W,H,M),
 \end{align}
and the entry-wise representation is as follows:
\begin{align}
	p(v_{ij}|\hat{v}_{ij}) = e^{-\hat{v}_{ij}}\frac{\hat{v}_{ij}^{v_{ij}}}{\Gamma(v_{ij}+1)},
\end{align}
where $\hat{v}_{ij} = \sum_{k=1}^{K}(w_{ik}h_{kj} - w_{ik}m_{ik} h_{kj}m_{jk})$ according to Eq.   (\ref{CP-NMF-EQ}). %
The corresponding negative log-likelihood is:
 \begin{align}
 	-{\rm log}p(v_{ij}|\hat{v}_{ij}) =& -v_{ij}{\rm log}\hat{v}_{ij} + \hat{v}_{ij} + {\rm log}v_{ij}! \label{post-v-EQ}\\
 	=& v_{ij} {\rm log}(\frac{v_{ij}}{\hat{v}_{ij}}) + \hat{v}_{ij} + \kappa_1,\label{neglog-v}
 \end{align}
where the second equation is obtained by choosing an appropriate constant $\kappa$, given $V$ as a constant.

Elements of $W$ and $H$, namely, $(w_{ik}$ and $(h_{ik}$ are defined to follow a half-normal distribution 
\begin{align}
		p(w_{ik}|\beta_{k}) = \mathcal{HN}(w_{ik}|0,\beta_{k}^{-1}), \\
		p(h_{kj}|\beta_{k}) = \mathcal{HN}(h_{kj}|0,\beta_{k}^{-1}),
\end{align}
 where the probability density function of $\mathcal{HN}$  is as follows
 \begin{align}
 	\mathcal{HN}(x|0,\beta^{-1}) = \sqrt{\frac{2}{\pi}}\beta^{-\frac{1}{2}}exp(-\frac{1}{2}\beta x^2).
 \end{align}
 
 Thus the negative log likelihood of $W$ and $H$ is:
 \begin{align}
 		-{\rm log}p(\textbf{W}|\beta) &= \sum_{k=1}^{K}\sum_{i=1}^{N}\frac{1}{2}\beta_{k}w_{ik}^{2}-\frac{N}{2}{\rm log}\beta_{k},\\
 		-{\rm log}p(\textbf{H}|\beta) &= \sum_{k=1}^{K}\sum_{j=1}^{N}\frac{1}{2}\beta_{k}h_{kj}^{2}-\frac{N}{2}{\rm log}\beta_{k}\label{neglog-WH}.
 \end{align}
 
In addition, we consider $\beta_{k}$ as a value randomly drawn from a Gamma distribution with parameters $a$ and $b$
 \begin{align}
 	p(\beta_{k}|a_{k},b_{k}) = \frac{b_{k}^{a_{k}}}{\Gamma(a_{k})}\beta_{k}^{a_{k}-1}exp(-\beta_{k}b_{k}), 
 \end{align}
 and thus, the negative log-likelihood is:
 \begin{align}
 	-{\rm log}p(\beta_{k}|a,b)=\sum_{k=1}^{K}(\beta_{k}b-(a-1){\rm log\beta_{k}})+\kappa_2\label{neglog-beta},
 \end{align}
where $\kappa_2$ can be  another constant with an appropriate value. %

The mask matrix $M$ is to filter out connections between periphery nodes. %
As each entry of $M$ is non-negative and is less than one, we %
consider $m_{ik}$ follows a truncated normal distribution \cite{burkardt2014truncated}, with parameter $\mu$  and $\sigma^2$ as the mean and the variance respectively of the original normal distribution before truncation, and $[0,1]$ as the sample space after truncation
\begin{align}
	\label{likelihood-m}
	p(m_{ik}|\mu_{k},\overline{\sigma},0,1) &= \mathcal{TN}(\mu_{k},\overline{\sigma},0,1). 
\end{align}

 Formally, $\mathcal{TN}$ is defined as follows
 \begin{align}
	\label{truncated-normal-distribution}
	\mathcal{TN}(x|\mu,\sigma,0,1) &= \frac{1}{\sigma} \cdot \frac{\phi(\frac{x-\mu}{\sigma})}{\Phi(\frac{1-\mu}{\sigma})-\Phi(\frac{-\mu}{\sigma})},
 \end{align}
where 
\begin{align}
	\label{density}
	\phi(\delta) &= \frac{1}{\sqrt{2\pi}}exp(-\frac{1}{2}\delta^2), \\ 
	\label{cumulative}
	\Phi(\delta) &= \frac{1}{2}(1 + erf(\frac{\delta}{\sqrt{2}})),\\ 
	\label{erf-1}
	erf(\delta) &= \frac{2}{\sqrt{\pi}}\int_{0}^{\delta} e^{-t^2} dt \\ 
	\label{erf-2}
	&= \frac{2}{\sqrt{\pi}}\sum_{0}^{\infty}\frac{(-1)^{n}\delta^{2n+1}}{n!(2n+1)}  \\ 
	\label{erf-3}
	&\approx \frac{2}{\sqrt{\pi}} \sigma,
\end{align}
where Eq.  (\ref{erf-2}) is obtained by Taylor expansion, and the approximation in Eq.  (\ref{erf-3}) is obtained by taking the first term of the expansion.

Thus, by summarizing Eq.  (\ref{likelihood-m}), Eq.  (\ref{truncated-normal-distribution}), Eq.  (\ref{density}), Eq.  (\ref{cumulative}), Eq.  (\ref{erf-1}), Eq.  (\ref{erf-2}), and Eq.  (\ref{erf-3}) we can have 
 \begin{align}
 	\nonumber
	p(m_{ik}|\mu_{k},\overline{\sigma},0,1) %
	\nonumber\\
	&= \frac{1}{\overline{\sigma}} \cdot \frac{\phi(\frac{m_{ik}-\mu_{k}}{\overline{\sigma}})}{\Phi(\frac{1-\mu_{k}}{\overline{\sigma}})-\Phi(\frac{-\mu_{k}}{\overline{\sigma}})} \\
	\nonumber
	&\approx \frac{1}{\overline{\sigma}}\cdot \frac{\frac{1}{\sqrt{2\pi}\overline{\sigma}}exp(-\frac{1}{2}(\frac{m_{ik}-\mu_{k}}{\overline{\sigma}})^2)}
	{\frac{1}{\sqrt\pi}(\frac{1-\mu_{k}}{ \sqrt2 \overline{\sigma}})-\frac{1}{\sqrt\pi}(\frac{-\mu_{k}}{\sqrt2 \overline{\sigma}})} \\
	&=\frac{1}{\overline{\sigma}}e^{\frac{-(m_{ik}-\mu_{k})^2}{2\overline{\sigma}^2}},
\end{align}
and  the negative log likelihood of $M$ is:
\begin{align}
	{\rm -log}p(M|\mu,\overline{\sigma}) = \sum_{k=1}^{K}\sum_{i=1}^{N}\frac{(m_{ik}-\mu_k)^2}{2\overline{\sigma}^2}+\overline{\sigma}\label{neglog-m},
\end{align}
where $\mu$ is a vector containing $\mu_k$ over different choices of  $k$s,  $\overline{\sigma}$ is a hyperparameter. %

In addition, we also consider $\mu_k$ follows a normal distribution with predefined hyper-parameters: mean $\hat{\mu}$ and variance $\hat{\delta}^2$ %
\begin{align}
	&p(\mu_k|\hat{\mu},\hat{\delta})=\mathcal{N}(\hat{\mu},\hat{\delta}^2),
\end{align}
and the negative log-likelihood is:
\begin{align}
	{\rm -log}P(\mu_k|\hat{\mu},\hat{\sigma}) = \frac{(\mu_{k}-\hat{\mu})^2}{2\hat{\sigma}^2} + \kappa_3. \label{neglog-mu}
\end{align}

By Bayesian theorem, to find the best choice of $W$, $H$, $\beta$, $M$,  and $\mu$, we can optimize the posterior distribution as follows
\begin{align}
	\nonumber
	&p(W,H,M,\beta,\mu | V) \\ 
	=&\frac{p(V|W,H,M)p(W|\beta)p(H|\beta)p(M|\mu)p(\beta)p(\mu)}{p(V)},
\end{align}
which is equivalent to minimize the following negative log likelihood %
\begin{align}
	\label{lossfunc}
	\nonumber
	\mathcal{U} =& -{\rm log}P(V|W,H,M)  -{\rm log}P(W|\beta) -{\rm log}P(H| \beta)\\
	& - {\rm log}P(M| \mu) - {\rm log}P(\mu) - {\rm log}P(\beta)\\ %
	=& \sum_{i=1}^{N}\sum_{j=1}^{N}(v_{ij} \cdot {\rm log}\frac{v_{ij}}{\hat{v}_{ij}} +\hat{v}_{ij}) \nonumber \\
	\nonumber
	&+\sum_{k=1}^{K}\sum_{i=1}^{N}(\frac{1}{2}\beta_{k}w_{ik}^{2})-\frac{N}{2}{\rm log}\beta_{k}\\
	\nonumber
	&+\sum_{k=1}^{K}\sum_{j=1}^{N}(\frac{1}{2}\beta_{k}h_{kj}^{2})-\frac{N}{2}{\rm log}\beta_{k} \\
	\nonumber
	&+\sum_{k=1}^{K}(\beta_{k}b-(a-1){\rm log\beta_{k}}) \\
	\nonumber
	&+\sum_{k=1}^{K}\sum_{i=1}^{N} \frac{(m_{ik}-\mu_k)^2}{2\overline{\sigma}^2}\\
	&+\sum_{k=1}^{K}\frac{(\mu_{k}-\hat{\mu})^2}{2\hat{\sigma}^2}
	+\kappa,
\end{align}
where the last equation is obtained by substituting corresponding terms with Eq.   (\ref{neglog-v}), Eq.  (\ref{neglog-WH}), Eq.  (\ref{neglog-beta}), Eq.   (\ref{neglog-m}), and Eq.   (\ref{neglog-mu}).

$K$ is usually chosen to be large enough. When converged, most columns of $W$ and rows of $H$ are likely to be zeros due to their prior settings, and the remaining can indicate  identified core-periphery pairs. %

Thus, $W_{ik}$ or $H_{ki}$ can indicate how likely node $i$ belong to core-periphery pair $k$, and $1-M_{ik}$ can indicate the core score of node $i$ in pair $k$. %
The continuous output provides a quantitative metric to measure the importance of each node in different pairs, with a "soft" core-periphery structure identification, which can be useful in identifying  overlapping pairs.

To compare with many traditional approaches, we can discretize the output of our approach. For non-overlapping pair detection, we can choose $\arg\max_k W_{ik}$ as the pair affiliation of node $i$. 
Similarly, for identify core nodes explicitly in a given pair $k$, we can choose node $i$ as a core node if $M_{ik}$ is less than the average value of $M_{:k}$.  %

\subsection{Optimization method}

In this part, we propose a multiplicative approach for optimization. 
The goal of optimization is to find appropriate $W$, $H$, $M$, $\beta$, and $\mu$ so that the objective function, Eq.  ({\ref{lossfunc}}), can be minimized. Inspired by the approach proposed by Lee and Seung  \cite{lee2000algorithms}, in this work, we propose a multiplicative approach for masked Bayesian non-negative matrix factorization. 
The update rule can be deducted as follows. 

First, we consider the gradient descent approach %
\begin{align}
	W^{*} &= W + \eta_{W}(\nabla_{W}\mathcal{U})\label{update-W}\\
	H^{*} &= H + \eta_{H}(\nabla_{H}\mathcal{U})\label{update-H},\\
	M^{*} &= M + \eta_{M}(\nabla_{M}\mathcal{U})\label{update-M}.
\end{align}

Taking $W$ as an example, the gradient is 

\begin{align}
	\nabla_{W}\mathcal{U}=
	& -\frac{V}{\hat{V}} H^T + (\frac{V}{\hat{V}} (H \circ M^T)^T) \circ M \nonumber \\ 
	& + (\textbf{1}  H^T - (\textbf{1} (H \circ M^T))\circ M) \nonumber \\ 
	& + W  B,\label{wGrad}
\end{align}
where $\hat{V} = W H - (W \circ M)(H \circ M^T)$, $B = \beta I$, $I$ indicates identity matrix, $\textbf{1} \in \mathbb{R}^{N \times N}$ is an all-one matrix, and '$\circ$' indicates element-wise matrix multiplication.
By purposely choosing an appropriate step length as follows

\begin{align}
	\eta_{W} = - \frac{W}{(\textbf{1}  H^T - (\textbf{1} (H \circ M^T))\circ M) + W  B},
\end{align}
the negative terms in Eq.  (\ref{wGrad}) can be eliminated, and the updated $W$ can be non-negative when approaching the optimum. 

Similarly, the gradient and the step length of $H$ and $M$ are as follows:
\begin{align}
	\nonumber
	&\nabla_{H}\mathcal{U} \\
	\nonumber
	=& -W^T  \frac{V}{\hat{V}} + M^T \circ ((W \circ M)^T  \frac{V}{\hat{V}})\\
	& +(W^T  \textbf{1} - ((WM)^T  \textbf{1}) \circ M^T) + B  H,\\
	\nonumber
	&\nabla_{M}\mathcal{U} \\
	\nonumber
	=& H^T \circ (\frac{V}{\hat{V}} W \circ M) + W \circ (\frac{V}{\hat{V}} H^T \circ M)\\
	\nonumber
	& - H^T \circ (\textbf{1}  W \circ M) - W \circ (\textbf{1}  H^T \circ M) \\
	& + \frac{(M-\mu)_{-}}{\overline{\sigma}^2}+ \frac{(M-\mu)_{+}}{\overline{\sigma}^2},\\
	\nonumber
	&\eta_{H}= \\
	 -& \frac{H}{(W^T \textbf{1} - ((W \circ M)^T \textbf{1}) \circ M^T) + B  H},\\
	\nonumber
	&\eta_{M}= \\ 
	 -& \frac{M}{(\frac{V}{\hat{V}}(H \circ M^T)^T)\circ W
		+ (\frac{V}{\hat{V}}(W \circ M))\circ H^T 
		+ \frac{(M-\mu)_{+}}{\overline{\sigma}^2} },
\end{align}
where the operations $(\cdot)_{+}$ and $(\cdot)_{-}$ denote to keep the positive and the negative entries of the matrix, respectively, by zeroing out other entries. 
With gradients and 
step lengths, Eq. (\ref{update-W}), (\ref{update-H}), and (\ref{update-M}) can be expressed as follows: 
\begin{align}
&	W^{*} =\nonumber\\& \frac{W \circ (\frac{V}{\hat{V}}  H^T - (\frac{V}{\hat{V}}  (H\circ M^T)^T) \circ M )}{(\textbf{1}  H^T - (\textbf{1} (H \circ M^T)) \circ M) + W  B},\\
&	H^{*} =\nonumber\\& \frac{H \circ (W^T  \frac{V}{\hat{V}} + M^T \circ ((W \circ M)^T  \frac{V}{\hat{V}}))}{(W^T \textbf{1} - ((W \circ M)^T \textbf{1}) \circ M^T) + B  H},\\
&	M^{*} =\nonumber\\& \frac{M \circ(H^T \circ (\textbf{1}  (W \circ M)) + W \circ (\textbf{1} (H^T \circ M)) - \frac{(M-\mu)_{-}}{\overline{\sigma}^2})}{(\frac{V}{\hat{V}}(H \circ M^T)^T)\circ W
		+ (\frac{V}{\hat{V}}(W \circ M))\circ H^T 
		+ \frac{(M-\mu)_{+}}{\overline{\sigma}^2} }.
\end{align}

During the iteration process, $W$, $H$, $M$, $\beta$, and $\mu$ are initialized with non-negative values, and $M$ is always projected into the interval of [0, 1] after each update, 
so that,  $W^*$, $H^*$ and $M^*$ are always non-negative.

The local optimum of
$\beta$ and $\mu$ can be directly solved by setting the gradient of  Eq. (\ref{lossfunc}) to  be zero, with respect to each variable. 

\begin{align}
	\nabla_{\beta_{k}^*}\mathcal{U} =& \frac{1}{2}(\sum_{i=1}^{N}w_{ik}^2 + \sum_{j=1}^{N}h_{jk}^2)-\frac{N+a-1}{B^*}+b = 0,\\
	\nabla_{\mu^*}\mathcal{U} =& \frac{\sum_{i=1}^{N}(\mu_k^*-m_{ik})}{\overline{\sigma}^2} + \frac{(\mu_{k}^*-\hat{\mu})}{\hat{\sigma}^2} = 0.
\end{align}

Thus, in each iteration, it can be  computed by solving the above equations, and the update rules for $\beta$ and $\mu$ are:
\begin{align}
	\label{update-mu}
	\mu_{k}^* &= \frac{\hat{\sigma}^2}{N\hat{\sigma}^2 + \overline{\sigma}^2}\sum_{i=1}^N m_{ik} 
	+ \frac{\overline{\sigma}^2}{N\hat{\sigma}^2 + \overline{\sigma}^2}\hat{\mu}, \\
	\label{update-b}
	\beta_{k}^* &= \frac{N + a - 1}{\frac{1}{2}(\sum_{i=1}^{N}w_{ik}^2 + \sum_{j=1}^{N}h_{jk}^2) + b}.
\end{align}

The overall optimization approach is as shown in Algorithm.\ref{algo}. %

\begin{algorithm}[!t]
	\caption{Core-Periphery Detection}
	\begin{algorithmic}[1]\label{algo}
		\ENSURE The observed adjacency matrix $V \in \mathbb{R}^{N \times N}_{+}$ 
		\ENSURE Hyperparameters $a, b, \overline{\sigma}, \hat{\sigma}, \hat{\mu}$
		\ENSURE Initial numbers of core-periphery-pairs $K$ 
		\ENSURE Maximum iteration $n_{iter}$
		\REQUIRE Non-negative matrices $W,H,M$
		\REQUIRE Non-negative vectors $\beta,\mu$
		\STATE Initialize $W,H,M$ with random non-negative values
		\STATE Initialize $\beta, \mu$ with an all-one vector with dimension $K$
		\FOR{$i = i:n_{iter}$} 
		\STATE Update $W$ according to Eq.  (\ref{update-W})
		\STATE Update $H$ according to Eq.  (\ref{update-H})
		\STATE Update $M$ according to Eq.  (\ref{update-M})
		\STATE Update $\mu$ according to Eq.  (\ref{update-mu})
		\STATE Update $\beta$ according to Eq.  (\ref{update-b})
		\ENDFOR
		\RETURN $W,H,M,\beta,\mu$
	\end{algorithmic} 
\end{algorithm}

\subsection{Convergence Analysis}
\label{convergence_analysis}
In this part, we introduce the proof of the convergence of  Algorithm.\ref{algo}. Following the idea of the traditional procedure \cite{lee2000algorithms}, we can prove the convergence of $W$, $H$, and the mask matrix $M$.

\begin{definition}
	$G(h, h')$ is an auxiliary function for $F(h)$ if the following conditions can be satisfied:
	\begin{align}
		G(h, h') &\geq F(h), \\
		 G(h,h) &= F(h).
	\end{align}
\end{definition}

According to \cite{lee2000algorithms}, the auxiliary function can be useful under the following lemma:
\begin{lemma}
	If $G$ is an auxiliary junction, then $F$ is nonincreasing under the update:
	\begin{align}
		\label{convergence}
		h^{(t+1)} = \mathop{\arg\min}_{h}  G(h, h^{(t)})
	\end{align}
\end{lemma}
\begin{proof}
	\begin{align}
		F(h^{(t+1)}) \le G(h^{(t+1)}, h^{(t)}) \le G(h^{(t)}, t^{(t)}) = F(h^{(t)})
	\end{align}
\end{proof}

By defining the appropriate auxiliary functions, the update rules in Algorithm.\ref{algo} follows from Eq. (\ref{convergence}). 

Since the hyperparameters $a,b,\hat{\mu},\hat{\sigma},\overline{\sigma}\nonumber$ in loss function in Eq.  (\ref{lossfunc}) are fixed, the loss function with respect to $W$ can be simplified by:
\begin{align}
	\nonumber
	\mathcal{U} 
	=&\sum_{i=1}^{N}\sum_{j=1}^{N}(v_{ij} \cdot {\rm log}\frac{v_{ij}}{\hat{v}_{ij}}+\hat{v}_{ij})\\
	&+\sum_{k=1}^{K}\sum_{i=1}^{N}(\frac{1}{2}\beta_{k}w_{ik}^{2}).
\end{align}

\begin{theorem}
	By updating $W,H$ and $M$ under the rules presented Eq.  (\ref{update-W}), Eq.  (\ref{update-H}), and Eq.  (\ref{update-M}), the objective function Eq.  (\ref{lossfunc}) is non-increasing and will converge into a locally optimal solution.
\end{theorem}

\begin{proof}
	Without loss of generality, we first demonstrate the convergence of $W$, and $H$ can be proved similarly as it is symmetric to $W$. The objective function with respect to $W$ can be written as follows, by omitting terms without $M$
	\begin{align}
		\label{objective-func-w}
		\mathcal{U} \left( W \right)
		 =&V\cdot \log \frac{V}{{WH-\left( W \circ M \right) \left( H \circ M^{T} \right)}}\nonumber\\
		&+WH-\left( W \circ M \right) \left( H \circ M^{T} \right)+\sum{\frac{1}{2}W^2B},
	\end{align}
	where the superscript of $W^2$ indicates an element-wise square operation. Here is the auxiliary function for $\mathcal{U}(W)$:
	\begin{align}
		\label{auxiliary-func-w}
		&G\left(W, W^{\left( t \right)} \nonumber\right) \\
		=&V\log \frac{V}{\hat{V}}+\hat{V}+\frac{1}{2}W^2B\nonumber\\
		&+W^{\left( t \right)}\left( \hat{V}+\frac{1}{2}W^{2}B \right) -W\log \frac{V}{\hat{V}}.
	\end{align}
	
	By subtracting Eq.  (\ref{objective-func-w}) from Eq.  (\ref{auxiliary-func-w}), we can have
	\begin{align}
		G\left(W, W^{\left( t \right)} \right) -\mathcal{U} \left( W \right) 
		&=W^{\left( t \right)}\left( V+\frac{1}{2}W^{2}B \right) -W\log \frac{\hat{V}}{V}.
	\end{align}
	
	Given that
	\begin{align}
		W^{\left( t \right)}\left( V+\frac{1}{2}W^{2}B \right) >W\log \frac{V}{\hat{V}}.
	\end{align}
	
	Thus we can infer that
	\begin{align}
		G\left( W, W^{\left( t \right)} \right) >\mathcal{U} \left( W \right).
	\end{align}
	
	The gradient of the auxiliary function with respect of $W_{ij}$ is
	\begin{align}
		\nonumber
		\frac{\partial G( W,W^{(t)} )}{\partial W_{ij}}
		=& \left[ W^{(t)}( (\textbf{1} H^{T} - (\textbf{1} (H \circ M^{T})^T) \circ M)+ WB ) \right.\\ 
		& \left. -W \circ ( \frac{V}{\hat{V}} H^{T} - ( \frac{V}{\hat{V}}  (H \circ M^{T})^T) \circ M)) \right]_{ij}.
	\end{align}
	
	When the above gradient is 0, we can get the iterative formula of $W$ as follows:
	\begin{align}
		W^{\left( t \right)}
		&=\left[ \frac{W \circ \left(\frac{V}{\hat{V}} H^{T} - (\frac{V}{\hat{V}} (H \circ M^{T})^T) \circ M \right)}
		{\left( \textbf{1} H^{T} - (\textbf{1} (H \circ M^{T})^T) \circ M \right) + WB} \right]_{ij}.
	\end{align}

	Next, we demonstrate  the convergence of $M$. By omitting irrelevant terms, the corresponding objective function with respect to $M$ is 
	\begin{align}
		\label{objective-func-m}
		&\mathcal{U} \left( M \right)\nonumber\\ 
		=&V\circ \log \frac{V}{\sum{WH-\left( W \circ M \right) \left( H \circ M^{T} \right)}}\nonumber\\
		&+\sum_{k=1}^K{\sum_{i=1}^N{\frac{\left( m_{ik}-\mu _k \right) ^2}{2\bar{\sigma}^2}}}.
	\end{align}
	
	The auxiliary function for $\mathcal{U}(m)$ is
	\begin{align}
		\label{auxiliary-func-m}
		\nonumber
		&G\left( M, M^{\left( t \right)} \right) \\ 
		\nonumber
		=&V \circ \log \frac{V}{\hat{V}}+\frac{\left( M-\mu \right)_- ^2}{2\bar{\sigma}^2}+M^{\left( t \right)}\left( \hat{V} +\frac{\left( M -\mu \right)_+ ^2}{2\bar{\sigma}^2} \right) \\ 
		&-V \circ \log \frac{V}{\hat{V}}-\frac{\left( M - \mu \right) ^2}{2\bar{\sigma}^2}.
	\end{align}
	
	By subtracting Eq.  (\ref{objective-func-m}) from Eq.  (\ref{auxiliary-func-m}), we get
	
		\begin{align}
		\nonumber
		&	G\left( M,M^{\left( t \right)} \right) -\mathcal{U} \left( M \right)\\ \nonumber
		=&M^{\left( t \right)}\left( \hat{V} +\frac{\left( M-\mu \right)_+ ^2}{2\bar{\sigma}^2} \right) \\ 
		&-M\left( \log \frac{V}{\hat{V}}-\frac{\left( M-\mu \right)_- ^2}{2\bar{\sigma}^2} \right).
	\end{align}

	Thus we can infer that
	\begin{align}
		M^{\left( t \right)}\hat{V} > M\log \frac{V}{\hat{V}}.
	\end{align}
	
	The gradient of the auxiliary function with respect of $M_{ij}$ is
	
	\begin{align}
		\nonumber
		&\frac{\partial G\left( M,M^{(t)} \right)}{\partial M_{ij}}\\ 
		\nonumber
		=&\left[M^{\left( t \right)}\left( \frac{V}{\hat{V}} \left( H^T \circ M \right) \circ W \right. \right. \\
		\nonumber
		&+ \left. \frac{V}{\hat{V}} (W \circ M) \circ H^{T}+\frac{\left( M-\mu \right)_+ ^2}{\bar{\sigma}^2} \right) \\ 
		\nonumber
		&- M\left( \textbf{1} \left( H^T \circ M \right) \circ W \right. \\
		&+ \left. \left. \textbf{1} (W \circ M) \circ H^{T}-\frac{\left( M-\mu \right)_- ^2}{\bar{\sigma}^2} \right) \right]_{ij}.
	\end{align}

	When the above gradient is 0, we can get the iterative formula of $M$ as follows:
	\begin{align}
		\nonumber
		&M^{\left( t \right)} =\\ 
		&  \frac{\left[ M \circ \left( \textbf{1} \left( H^T \circ M \right) \circ W + \textbf{1} (W \circ M) \circ H^{T}-\frac{\left( M-\mu \right)_{-} ^2}{\bar{\sigma}^2} \right) \right]_{ij}}
		{\left[ \left( \frac{V}{\hat{V}} \left( H^T \circ M \right) \circ W + \frac{V}{\hat{V}} (W \circ M) \circ H^{T}+\frac{\left( M-\mu \right)_{+} ^2}{\bar{\sigma}^2} \right) \right]_{ij}}.
	\end{align}
\end{proof}

\subsection{Time complexity}
Eq. (\ref{update-W}),  (\ref{update-H}) and  (\ref{update-M}) involves the multiplication of the two matrices shaped $\mathbb{R}^{N \times K}$ and $\mathbb{R}^{K \times N}$, so the overall computational complexity of Algorithm \ref{algo} is $\mathcal{O}( KN^2 )$.

\section{Experiment}

In this section, we demonstrate experimentally that %
our proposed method can effectively identify core-periphery structures. %

We use the evaluation method based on discrete core-periphery partitioning. In the context of community detection, a commonly used evaluation metric is Normalized mutual information (NMI)\cite{mcdaid2011normalized}, where the definition is given by:
\begin{align}
	NMI(Y,C) = \frac{2 \times I(Y;C)}{[H(Y) + H(C)]}.
\end{align}

Here, $Y$ is class labels, and $C$ is pair labels. $H(\cdot)$ is entropy and $I(Y;C)$ is the mutual information between $Y$ and $C$. The value of NMI is between 0 and 1, where 0 denotes no mutual information, and 1 denotes $Y$ and $C$ are identical. In \cite{shen2021finding}, the author further proposed a NMI metric for the core-periphery detection by considering the correctness of both pair classification and core edge classification, namely:
\begin{align}
	NMI_{cp} = \frac{1}{2}(NMI(r, \hat{r}) + NMI(c, \hat{c})),
\end{align}
where $r$ represents the true label of the core-periphery pair and $c$ represents the true classification label of the core and periphery nodes. The value of $NMI_{cp}$ is between 0 and 1, and a larger $NMI_{cp}$ means that the result of core-periphery partition is approximately close to the ground truth. We set $a$, $b$, $\overline{\sigma}$, $\hat{\sigma}$, $\hat{\mu}$, and $K$ as $5$, $10$, $1$, $1$, and $32$, by default. %

\subsection{Random networks with non-overlapping core-periphery pairs}\label{non-overlapping}

In this part, we measure our method on random networks with non-overlapping core-periphery structures. Block model  \cite{stochastic} have been widely used in complex network analysis.
Here, we adopt a similar approach proposed by Zhang et al. \cite{zhang2015identification} to generate synthetic networks with core-periphery structures.
The parameters including the proportion of core nodes in each core-periphery pair, the connection probability between core nodes or between core nodes and other nodes in the same pair, and that probability between other nodes are set to be 0.5, 0.6, and 0.6 respectively. %

We compare our method with five different algorithms on random networks, including core-periphery score maximization (CSM)\cite{shen2021finding}, Lap-Core (LC), LowRank-Core (LRC)\cite{cucuringu2016detection}, MINRES (MIN)\cite{minres}, and KM-config (KM)\cite{kojaku2017finding} algorithm. The $NMI_{cp}$ performances of different algorithms on random networks of different sizes are presented in Table \ref{benchmark}. For each value of $N$, we generated five random networks and obtained the core-periphery partition by the above six algorithms respectively, recorded the $NMI_{cp}$ of each experiment, and the result was taken as the average of these five experiments. The result shows that our model generally performs better than other methods, especially in large networks ($N \geq 2000$). 
To clearly illustrate the effectiveness of our method, we draw the core-periphery partition result for the case $N = 5000$ in Fig.  \ref{randomnetworks}.
Furthermore, we compare the runtime of our method with other methods, the results are shown in Fig. \ref{complexity}. Although our approach on CPU has a relatively high time complexity, 
when the matrix operations are accelerated by GPUs with Pytorch, the computing time can be significantly reduced. 

We also conduct a sensitivity analysis under a synthetic network of size $N = 5000$ on hyper-parameters, namely, $a$, $b$, $\overline{\sigma}$, $\hat{\sigma}$, $\hat{\mu}$, and $K$. For each hyperparameter, we apply different changes (-30\%, -20\%, -10\%, 0\%, 10\%, 20\%, 30\%) with respect to its predefined value, to study the sensitivity of the model for each hyperparameter. For each change of a hyperparameter, we conduct 10 independent experiments and record the average $NMI_{cp}$. The averaged results are shown in Table \ref{sensitivity_analysis}. We also calculate the mean and the standard deviation of results for each hyperparameter, 
as shown in Fig.\ref{sensitivity}. In general, our model can maintain good and stable performance under fluctuated hyperparameter values.

\begin{table}[!t]  
	\caption{$NMI_{cp}$ for different algorithms}
	\label{benchmark}
	\centering
	\begin{tabular}{lcccccc}
		
		\toprule
		\multirow{3}{*}{N} &\multicolumn{6}{c}{${\rm NMI}_{cp}$} \\
		\cline{2-7}\\
		&NMF &CSM & KM &LC &LRC &MIN \\
		\cline{1-7}\\
		1000  &0.430 &\textbf{0.485} &0.379 &0.115 &0.095 &0.089\\ 
		2000  &\textbf{0.616} &0.538 &0.522 &0.036 &0.080 &0.083\\ 
		3000  &\textbf{0.671} &0.536 &0.503 &0.019 &0.080 &0.079\\ 
		4000  &\textbf{0.697} &0.538 &0.518 &0.017 &0.077 &0.077\\ 
		5000  &\textbf{0.760} &0.529 &0.512 &0.006 &0.051 &0.075\\ 
		6000  &\textbf{0.791} &0.534 &0.513 &0.020 &0.028 &0.074\\ 
		7000  &\textbf{0.795} &0.538 &0.529 &0.005 &0.011 &0.073\\ 
		8000  &\textbf{0.851} &0.537 &0.528 &0.004 &0.007 &0.072\\ 
		9000  &\textbf{0.880} &0.544 &0.528 &0.006 &0.005 &0.071\\ 
		10000 &\textbf{0.863} &0.540 &0.527 &0.013 &0.004 &0.070\\
		\bottomrule  
		
	\end{tabular}
	
\end{table} 

\begin{table}[!t]  
	\caption{Sensitivity analysis}
	\label{sensitivity_analysis}
	\centering
	\begin{tabular}{rcccccc}
		
		\toprule
		\multirow{3}{*}{variance} &\multicolumn{6}{c}{${\rm NMI}_{cp} (N=5000)$} \\
		\cline{2-7}\\
		&$K$ &$a$ &$b$ &$\overline{\sigma}$ &$\hat{\sigma}$ &$\hat{\mu}$ \\
		\cline{1-7}\\
		$-30\%$  &0.812 &0.803 &0.807 &0.831 &0.821 &0.828\\ 
		$-20\%$  &0.802 &0.785 &0.802 &0.814 &0.822 &0.811\\ 
		$-10\%$  &0.806 &0.830 &0.828 &0.803 &0.827 &0.809\\ 
		$  0\%$  &0.779 &0.833 &0.837 &0.806 &0.782 &0.820\\ 
		$+10\%$  &0.829 &0.824 &0.812 &0.778 &0.821 &0.823\\ 
		$+20\%$  &0.804 &0.767 &0.825 &0.789 &0.831 &0.812\\ 
		$+30\%$  &0.809 &0.836 &0.830 &0.813 &0.815 &0.804\\
		\cline{2-7}\\
		$ave$  &0.806 &0.812 &0.821 &0.805 &0.817 &0.816\\ 
		$std$  &0.014 &0.025 &0.012 &0.016 &0.015 &0.008\\ 
		\bottomrule  
		
	\end{tabular}
	
\end{table} 

\begin{figure}[!t]
	\centering %
	\includegraphics[width=0.45\textwidth]{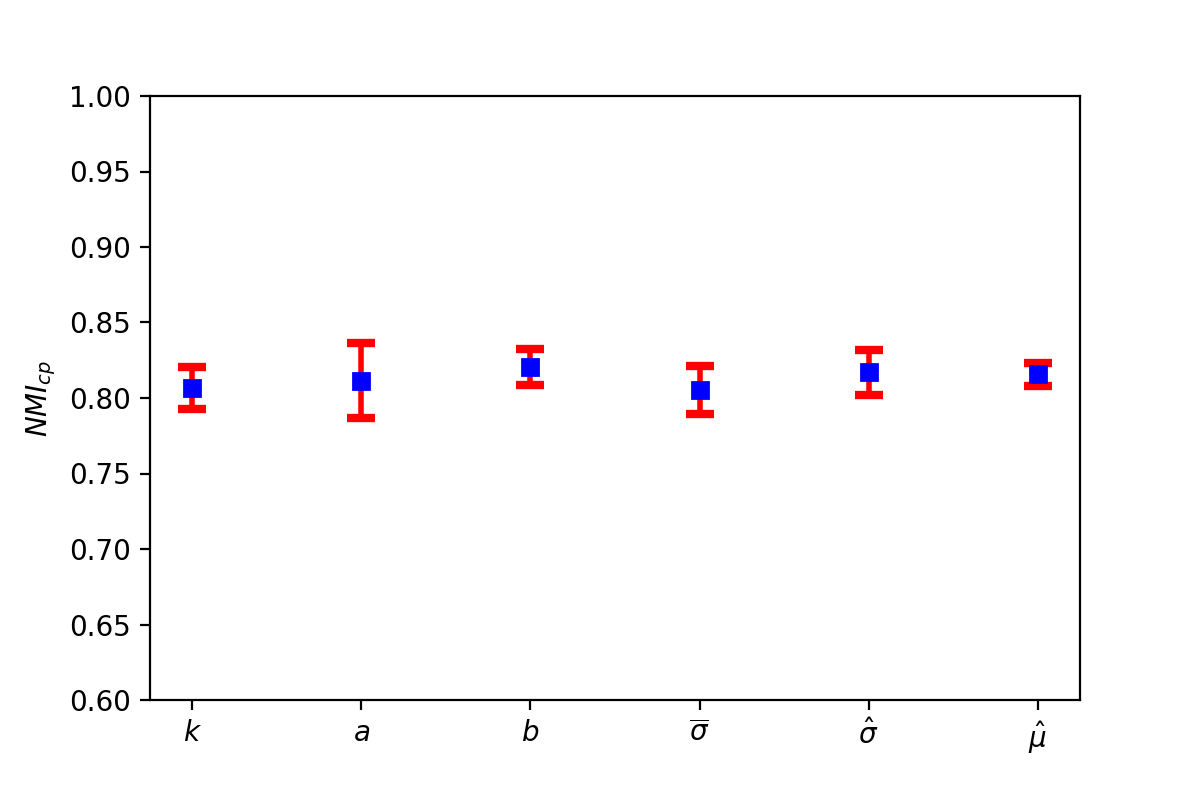} %
	\caption{
		Sensitivity analyze of  hyperparameters. For each hyperparameter, we apply a variation on it and run experiments to record  $NMI_{cp}$. We report the mean and standard deviation of the recorded results.
	}
	\label{sensitivity} %
\end{figure}%

\begin{figure}[!t]
	\centering %
	\includegraphics[width=0.45\textwidth]{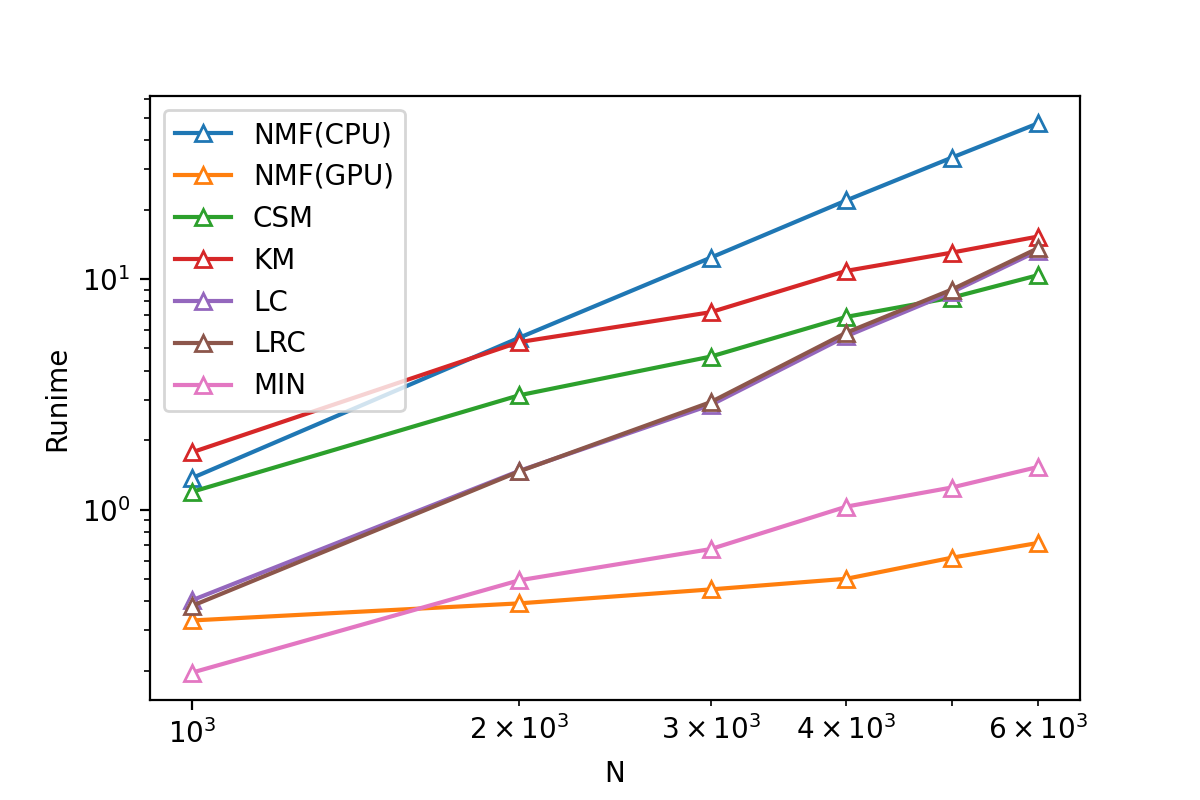} %
	\caption{Runtime of different algorithms} %
	\label{complexity} %
\end{figure}%

\begin{figure}[!t]
	\centering %
	\includegraphics[width=0.45\textwidth]{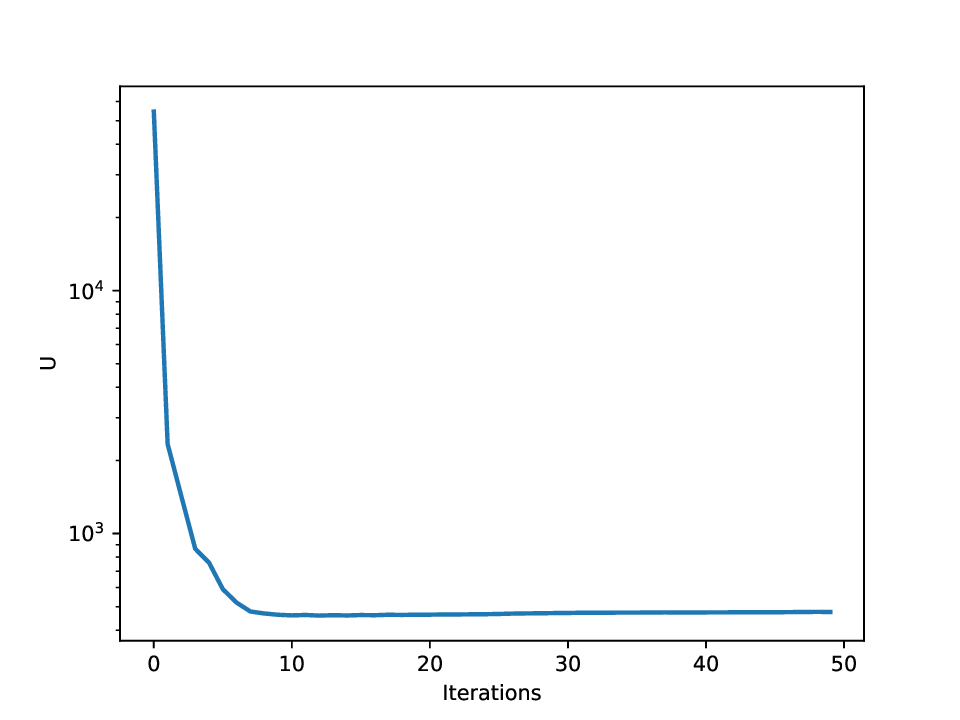} %
	\caption{
		Convergence analysis on Polbooks dataset 
	} %
	\label{convergence_experiment} %
\end{figure}%

\subsection{Random networks with overlapping core-periphery structures}

\begin{figure}[!t]
	\centering %
	\includegraphics[width=0.45\textwidth]{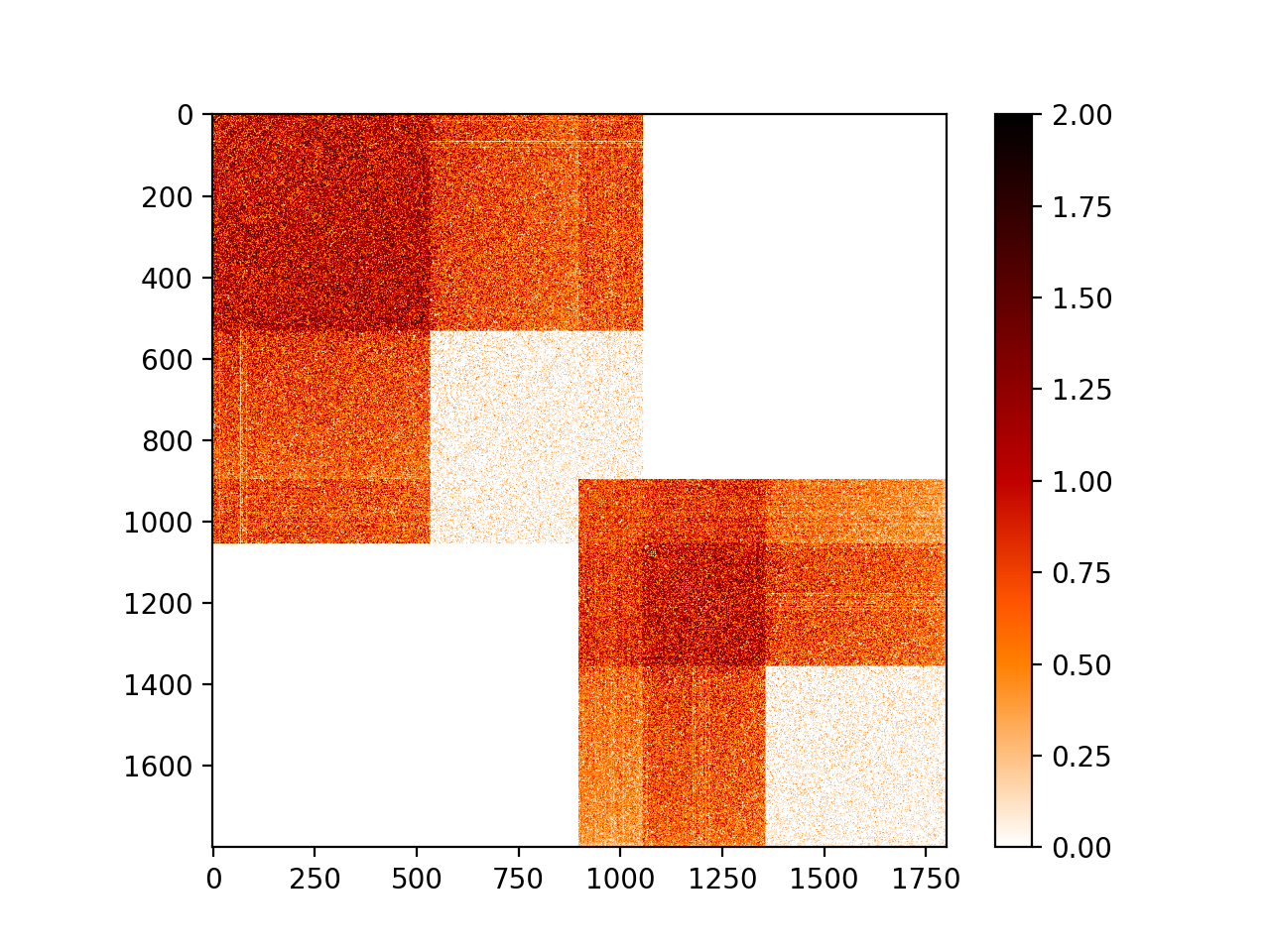} %
	\caption{Illustration of overlapping %
		core-periphery detection. Nodes are ordered with respect to $W$ and $M$, and the color intensity is obtained from $M$. %
		Darker areas indicate that corresponding nodes have higher core scores, and are more likely to be core nodes. %
		Two blocks represent two overlapping core-periphery pairs. The overlapping part in the middle indicates that, some nodes in the first pair are periphery nodes while those in the second pair are core nodes. %
	} %
	\label{overlap} %
\end{figure}%

In this section, we discuss the ability of our model to handle overlapping core-periphery structures. One common but often overlooked situation is when some node, $i$, is an periphery node in a core-periphery pair, but a core node in another core-periphery pair. Traditional methods tend to classify such a node as either a periphery node or a core node. %

Our approach can well address this problem by using the core score indicator $M$ for different pairs. To illustrate the overlapping solution, we generated a network with two core-periphery pairs based on the stochastic model described in the previous section and perform our approach. The experimental result has shown in Fig. \ref{overlap}. %
In this figure, nodes are ordered according to $W$ and $M$, and we also use a colorbar to represent the distribution of core score values. The experimental result demonstrates that our model can accurately locate core and periphery nodes in overlapping situations.

\subsection{Real-world networks}

In this section, we apply our algorithm to real-world networks and compare our algorithm with two other algorithms, i.e., CSM and KM-config. The dataset utilized is Polbooks\cite{polbooks}, Email-Eu-core\cite{email-eu-core1,email-eu-core2} and ego-Facebook\cite{ego-Facebook}. 

The Polbooks dataset is a network of books about U.S. politics published close to the 2004 U.S. presidential election and sold by Amazon.com.Each node represents a book, and edges between books represent frequent co-purchasing of those books by the same buyers. The network is comprised of 105 nodes and 441 edges. 

The email-Eu core network, which was generated using email data collected from a European research institution, is comprised of 1005 nodes and 25571 edges, where each node represents a person, each edge denotes that there is at least one email sent from one person to another correspondingly. %

The ego-Facebook network, which was collected from survey participants using this Facebook app, is comprised of 4039 nodes and 88234 edges. Each node represents a user, and each edge represents a social relation.

To evaluate our method, we reorder the adjacency matrix according to the output of different algorithms and display the sorting results in the form of Fig.  \ref{realnetwork-polbooks}, \ref{realnetwork-email}, and \ref{realnetwork-facebook}. We compare with traditional approaches CSM and KM by discretizing our output as described in the method section. 
We rearrange the adjacency matrix according to the algorithm output by grouping nodes within the same core-periphery pair, placing nodes in larger core-periphery pairs in front of those in smaller pairs, and arranging core nodes in front of the periphery nodes within the same core-periphery pair. 

In the ego-Facebook dataset, the overlapping between different core-periphery pairs is not very significant, so we use the same representation scheme as that for Fig. \ref{randomnetworks}, in which each red rectangle represent a core-periphery pair, and darker color indicate a higher core score. 
In Polbooks and Email-Eu-core datasets, there are significant  overlaps between different core-periphery pairs. When we use different colors to indicate different pairs respectively for results with NMF, it can be find that our approach can identify the overlapping correctly represented by mixed colors, while other approaches typically ignores such overlapping. 
These results show that, compared with traditional approaches, our approach has  advantages in identifying overlapping and non-overlapping core-periphery pairs. 

We conduct a numerical analysis to demonstrate the convergence rate on the Polbooks dataset, and the result is shown in Fig. \ref{convergence_experiment}. The experimental result indicates that after several iterations, the value of the objective function, Eq. (\ref{lossfunc}), gradually decrease to a constant.

\begin{figure}[htbp]

	\subfloat[NMF]{
		\centering
		\includegraphics[width=3.2in]{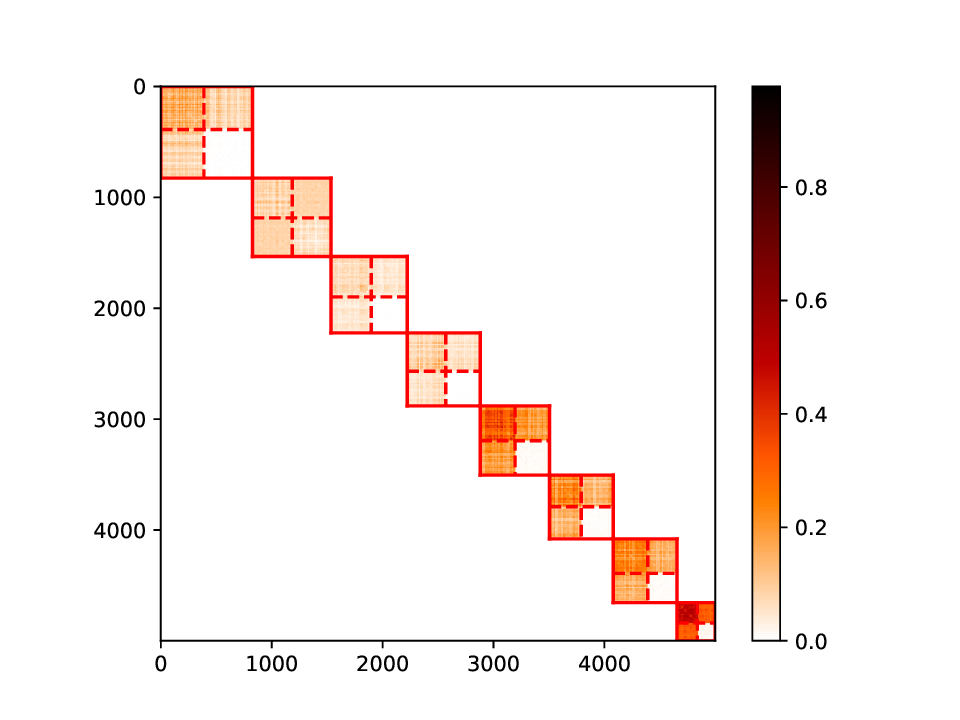}
	}\\
	\subfloat[CSM]{
		\centering
		\includegraphics[width=3.2in]{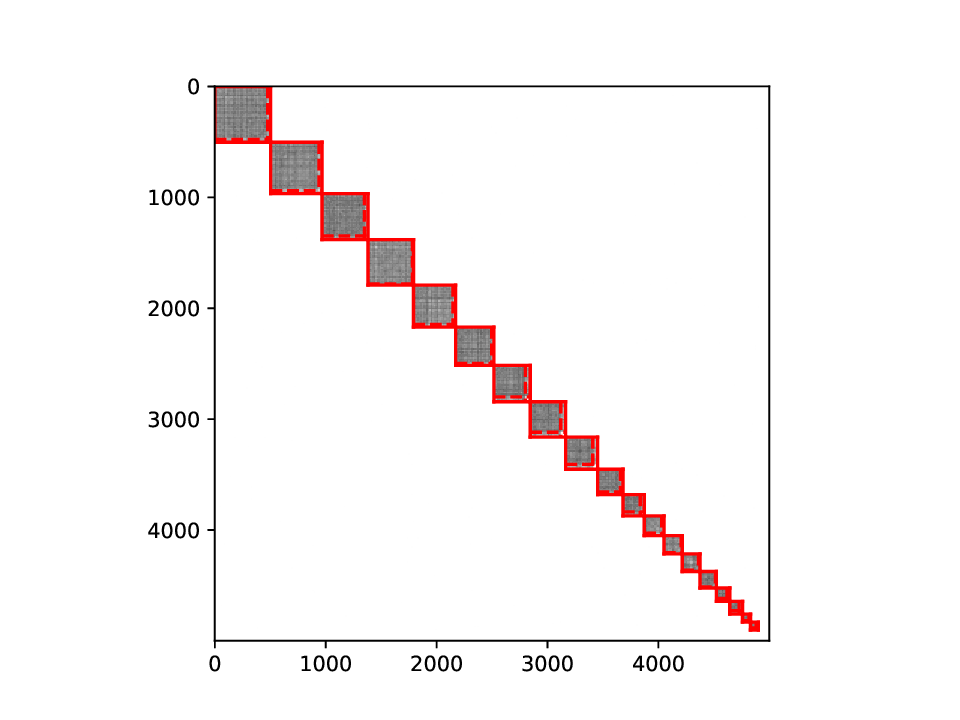}
	}\\
	\subfloat[KM]{
		\centering
		\includegraphics[width=3.2in]{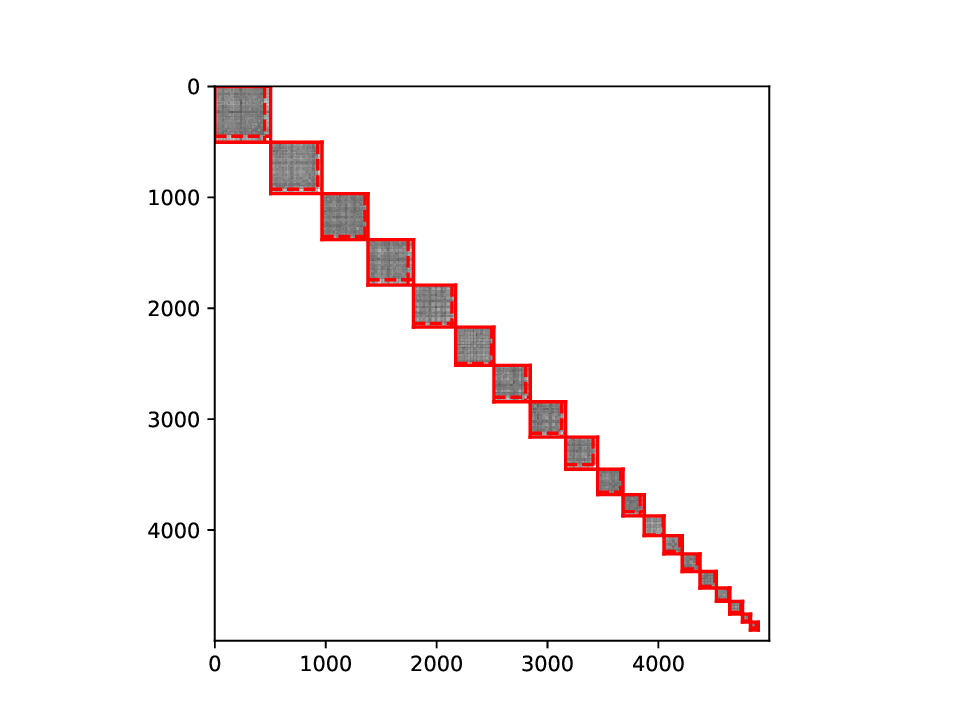}
	}\\
	\caption{Core-periphery detection on a random network ($N = 5000$) by Non-negative Matrix Factorization, core-periphery score maximization, and KM-config. Rectangles represent different core-periphery pairs, and the dashed lines divide the core and periphery regions.}
	\label{randomnetworks}
\end{figure}

\begin{figure}[htbp]
	\centering
	\subfloat[NMF]{
		\includegraphics[width=3.2in]{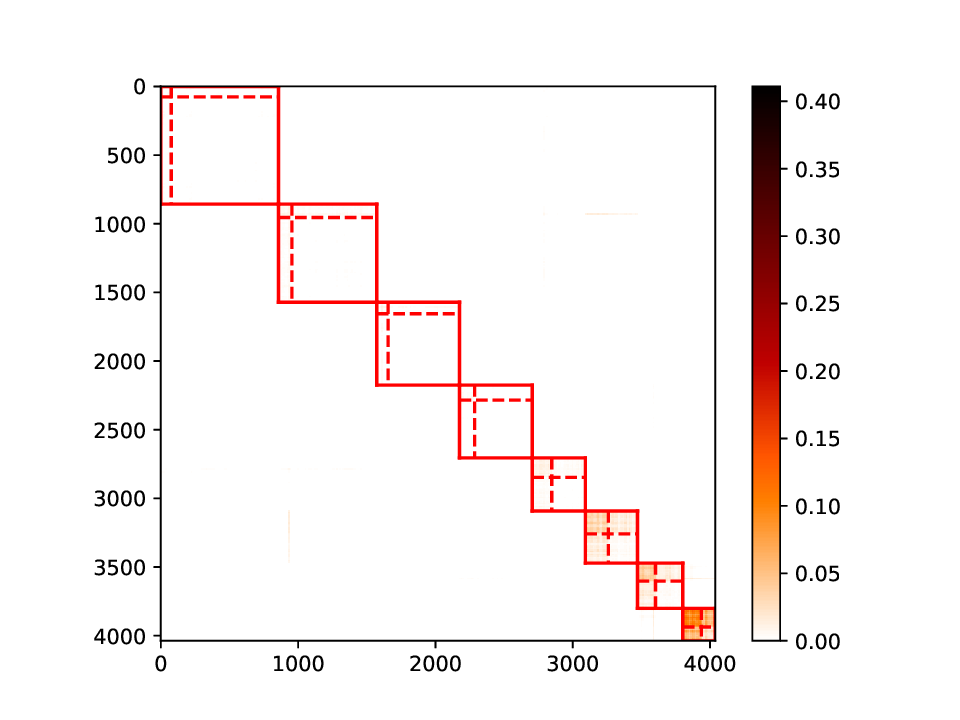}
	}\\
	\subfloat[CSM]{
		\includegraphics[width=3.2in]{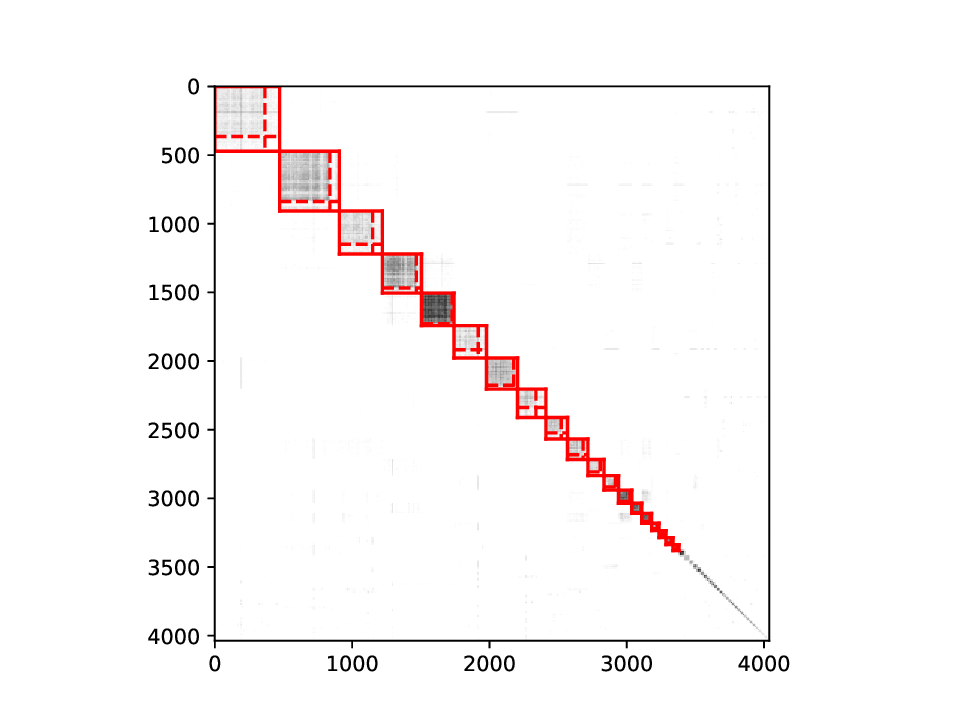}
	}\\
	\subfloat[KM]{
		\includegraphics[width=3.2in]{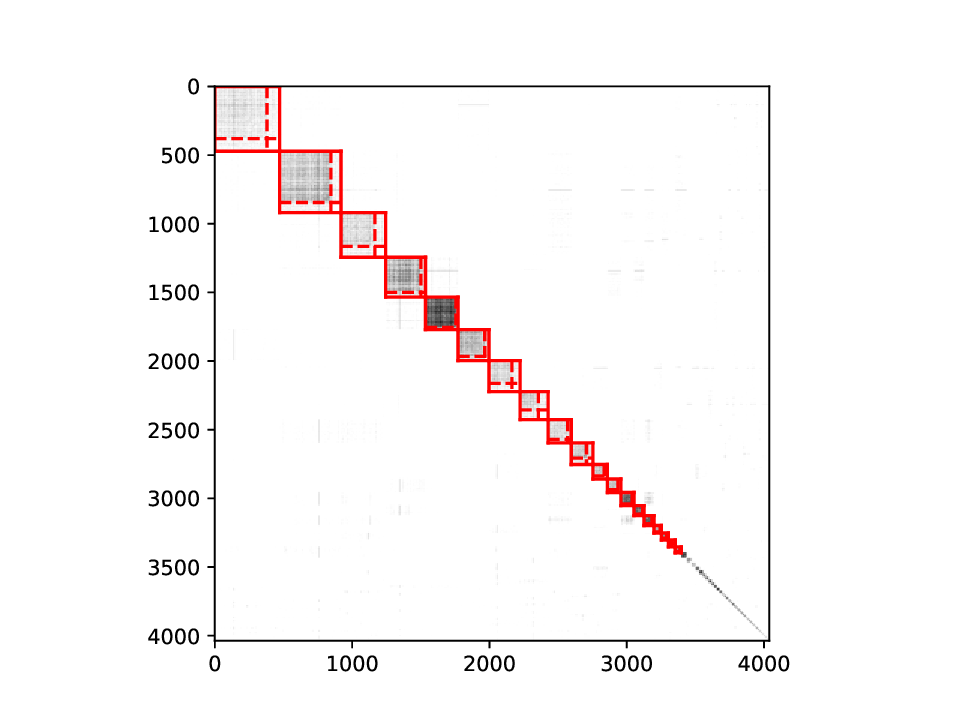}
	}\\
	\caption{Core-periphery detection on ego-Facebook by NMF, CSM and KM-config. Rectangles represent different core-periphery pairs, and the dashed lines divide the core and periphery regions.}
	\label{realnetwork-facebook}
\end{figure}

\begin{figure}[htbp]
	\centering
	\subfloat[NMF]{
		\includegraphics[width=3.2in]{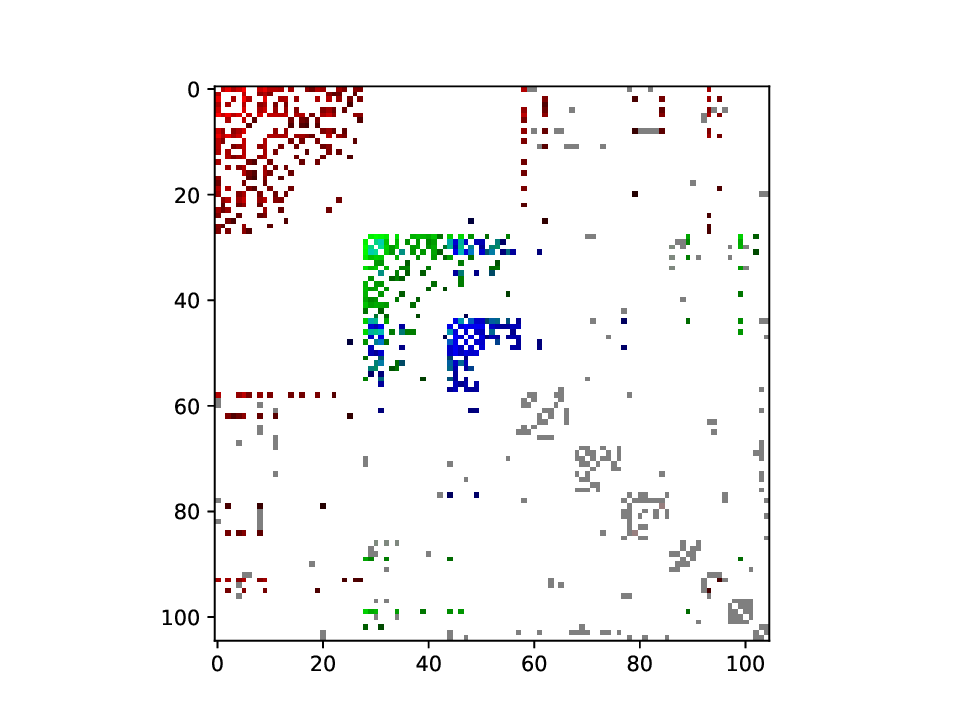}
	}\\
	\subfloat[CSM]{
		\includegraphics[width=3.2in]{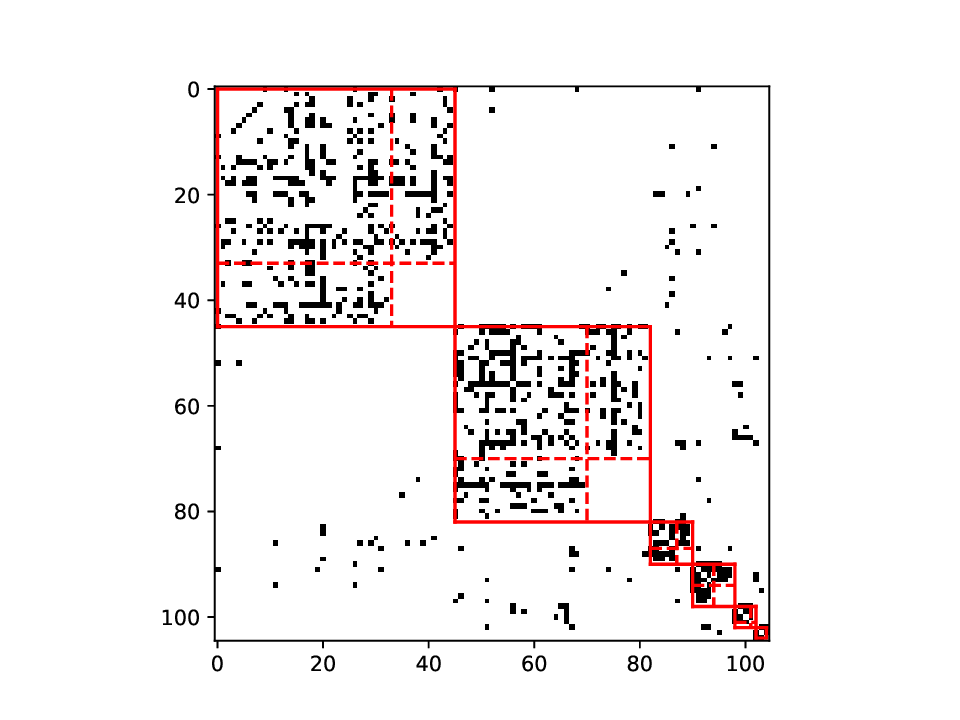}
	}\\
	\subfloat[KM]{
		\includegraphics[width=3.2in]{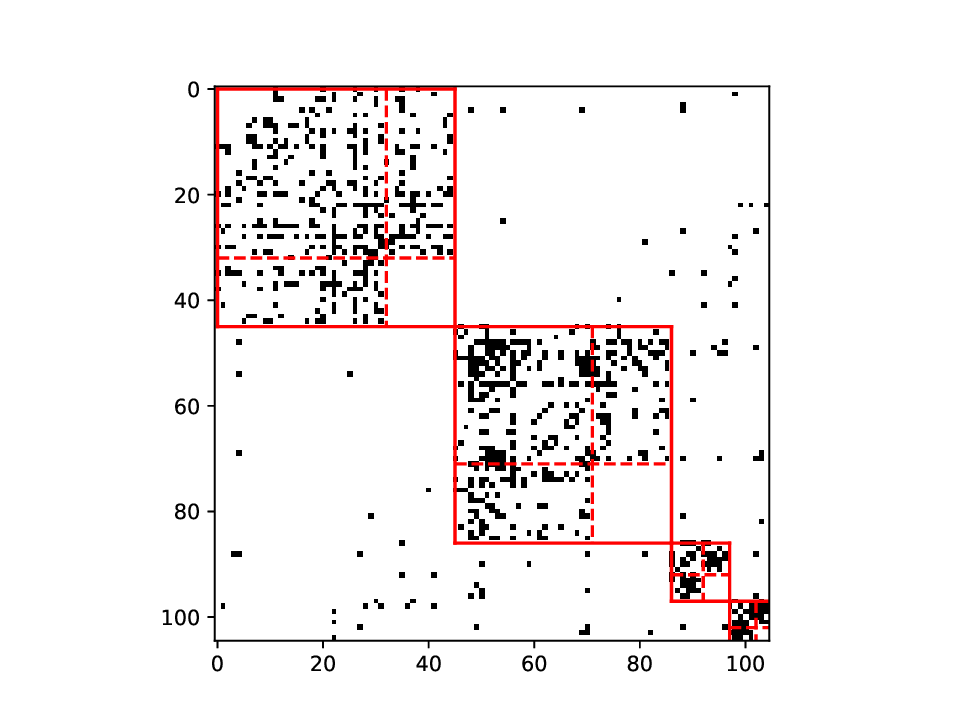}
	}\\
	\caption{
		The core-periphery detection results using NMF, CSM, and KM-config on the polbook dataset. In the result by  NMF, red, green, and blue colors represents three largest  core-periphery pairs respectively, and the mixing of these colors indicates core-periphery overlapping. In the  results by CSM and KM, rectangles represent different core-periphery pairs, with dashed lines separating the core and periphery regions.
	}
	\label{realnetwork-polbooks}
\end{figure}

\begin{figure}[htbp]
	\centering
	\subfloat[NMF]{
		\includegraphics[width=3.2in]{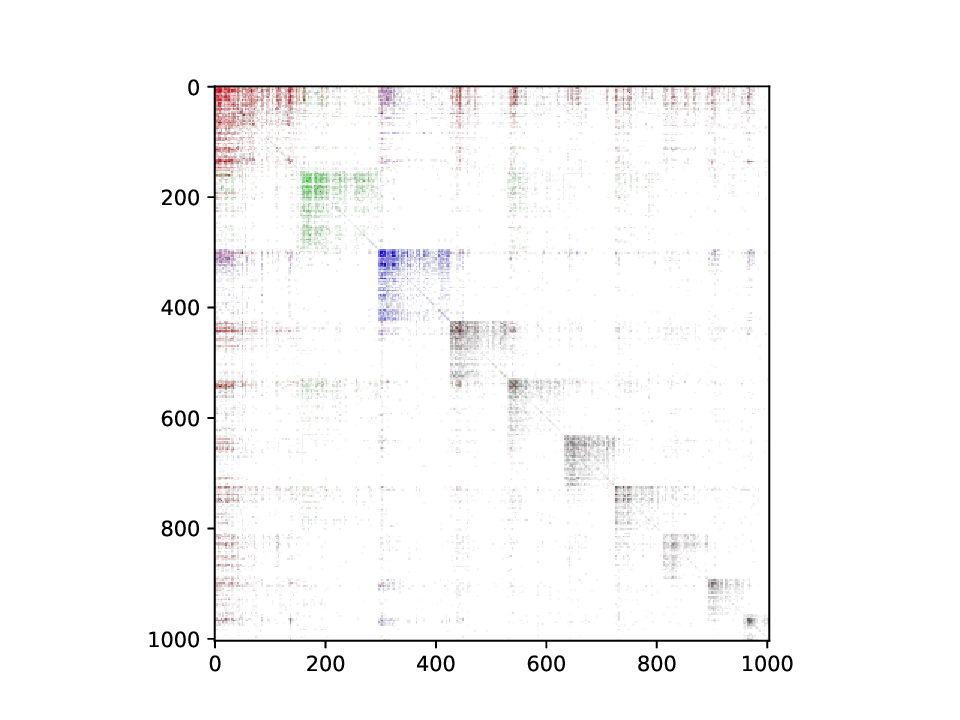}
	}\\
	\subfloat[CSM]{
		\includegraphics[width=3.2in]{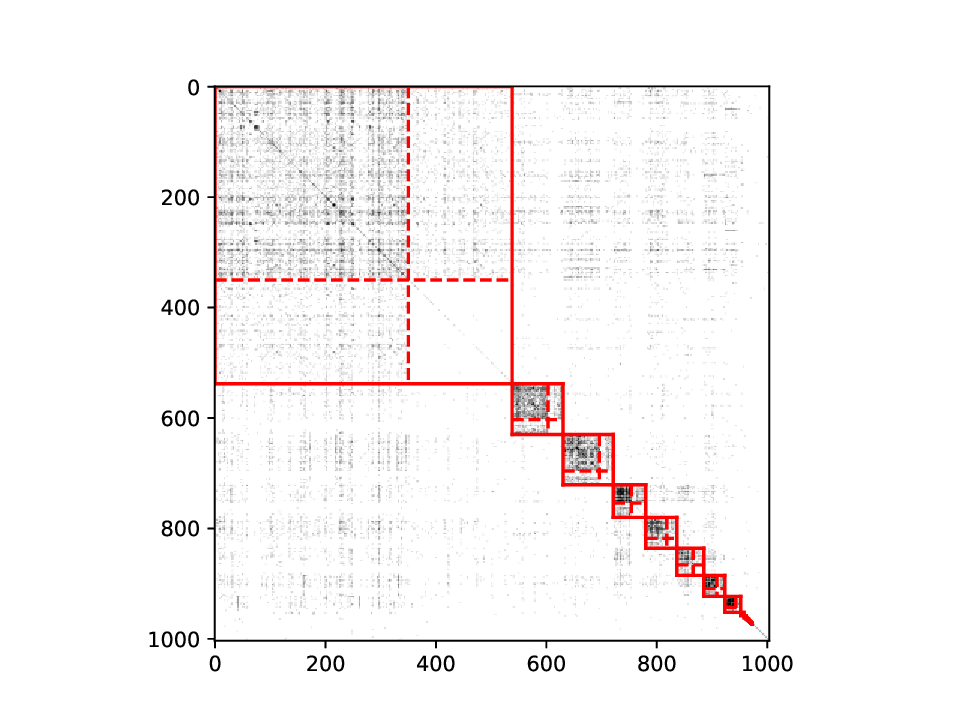}
	}\\
	\subfloat[KM]{
		\includegraphics[width=3.2in]{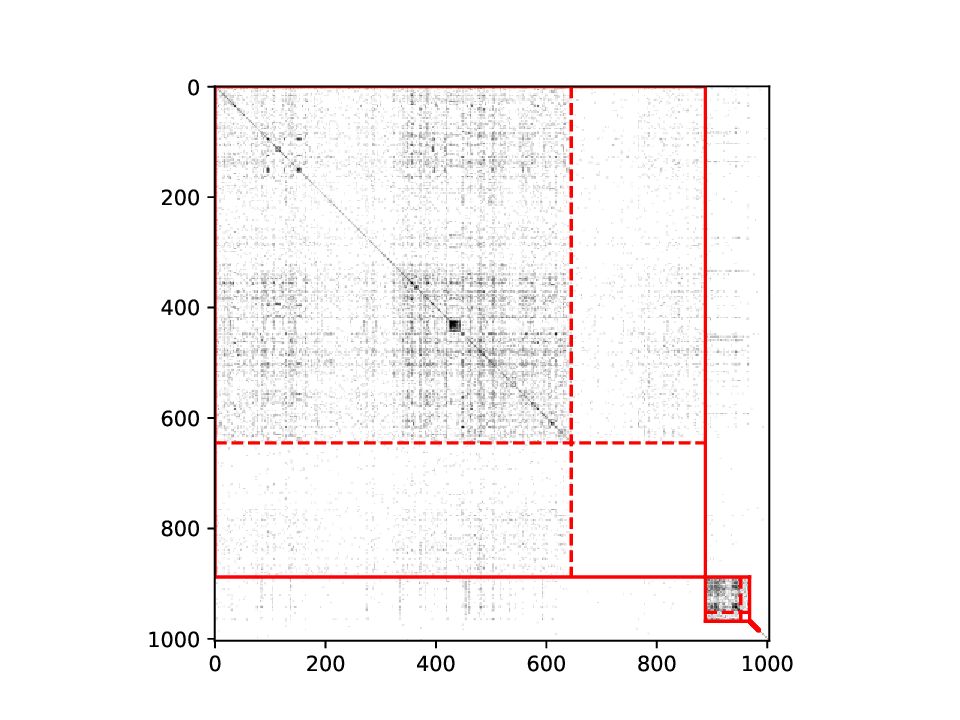}
	}\\
	\caption{
		The core-periphery detection results using NMF, CSM, and KM-config on the Email-Eu-core dataset. In the result by  NMF, red, green, and blue colors represents three largest  core-periphery pairs respectively, and the mixing of these colors indicates core-periphery overlapping. In the  results by CSM and KM, rectangles represent different core-periphery pairs, with dashed lines separating the core and periphery regions.
\label{realnetwork-email}
}
\end{figure}

\section{Conclusion}
In this paper, we propose a generative model called masked Bayesian non-negative matrix factorization, for detecting core-periphery structures. We propose an approach to optimize the model parameters and prove its convergence. Besides the ability of traditional approaches, our method can provide soft partitions and core scores, and it is capable to identify overlapping core-periphery pairs.
In the experiments, our approach can outperform traditional methods in different scenarios. Code of the proposed approach is available at https://github.com/HazwRuHi/Masked\_Bayesian\_NMF.

\bibliographystyle{IEEEtran}
\bibliography{mybib}

\end{document}